\documentclass[11pt]{article}

\usepackage[colorlinks=true,linkcolor=blue,citecolor=blue]{hyperref}

\usepackage{amsmath, amssymb, amsthm, thm-restate}
\usepackage{mathtools}
\usepackage[capitalize,nameinlink]{cleveref}
\usepackage{fullpage}
\usepackage{graphics}
\usepackage{pifont}
\usepackage{tikz}
\usepackage{multirow}
\usepackage{bbm}
\usepackage{cite}

\usetikzlibrary{arrows.meta}
\usepackage{environ}
\usepackage{framed}
\usepackage{url}
\usepackage{algorithm}
\usepackage[noend]{algpseudocode}
\usepackage[labelfont=bf]{caption}
\usepackage{cite}
\usepackage{framed}
\usepackage[framemethod=tikz]{mdframed}
\usepackage{appendix}
\usepackage{graphicx}
\usepackage[textsize=tiny]{todonotes}
\usepackage{tcolorbox}

\newcommand\remove[1]{}

\newtheorem{theorem}{Theorem}
\newtheorem{lemma}{Lemma}[section]
\newtheorem*{lemma*}{Lemma}

\newtheorem*{corollary*}{Corollary}

\newtheorem{definition}[lemma]{Definition}

\theoremstyle{definition}

\newtheorem*{theorem*}{Theorem}
\newtheorem{rem}[lemma]{Remark}

\newtheorem*{rem*}{Remark}
\newtheorem{example}{Example}

\newcommand\R{\mathbb{R}}

\newcommand\E{\mathbb{E}}

\newcommand{\eps}{\varepsilon}
\renewcommand{\O}{\widetilde{O}}

\newcommand{\bs}{\backslash}

\newcommand{\hf}{\hat{f}}

\newcommand{\assign}{\leftarrow}

\newcommand{\eo}{\mathrm{EO}}
\renewcommand{\forall}{\mathrm{\text{ for all }}}
\newcommand{\proj}{\mathrm{proj}}
\newcommand{\argmin}{\mathrm{argmin}}
\newcommand{\opt}{\mathrm{opt}}
\newcommand{\bg}{\mathrm{\textbf{g}}}
\newcommand{\bz}{\mathrm{\textbf{z}}}
\newcommand{\bd}{\mathrm{\textbf{d}}}
\newcommand{\Process}{\textsc{Process}}
\newcommand{\Sample}{\textsc{Sample}}
\newcommand{\Findperm}{\textsc{FindPerm}}
\newcommand{\SFM}{\textsc{SFM}}
\newcommand{\sign}{\mathrm{sign}}
\newcommand{\SparseSFM}{\textsc{SparseSFM}}
\newcommand{\FindIndex}{\textsc{FindIndex}}
\renewcommand{\eo}{\mathrm{EO}}
\newif\ifrandom
\randomtrue

\newcommand{\defeq}{\stackrel{\mathrm{\scriptscriptstyle def}}{=}}

\newcommand{\poly}{{\rm poly}}

\newcommand{\todolater}[1]{}

\newcommand{\etal}{\textit{et~al.}}

\author{
Brian Axelrod \\
Stanford University \\
\texttt{baxelrod@cs.stanford.edu}
\thanks{Research supported by in part by an NSF graduate research fellowship and a Finch family fellowship}
\and
Yang P. Liu \\
Stanford University \\
\texttt{yangpatil@gmail.com}
\thanks{Research supported by the U.S.
Department of Defense via an NDSEG fellowship.}
\and
Aaron Sidford \\
Stanford University \\
\texttt{sidford@stanford.edu}
\thanks{
Research was supported in part by NSF CAREER Award CCF-1844855.
}
}

\begin{document}

\title{Near-optimal Approximate Discrete and Continuous \\ Submodular Function Minimization}

\begin{titlepage}
\clearpage\maketitle
\thispagestyle{empty}

In this paper we provide improved running times and oracle complexities for approximately minimizing a submodular function. Our main result is a randomized algorithm, which given any submodular function defined on $n$-elements with range $[-1, 1]$, computes an $\eps$-additive approximate minimizer in $\tilde{O}(n/\eps^2)$ oracle evaluations with high probability. This improves over the $\tilde{O}(n^{5/3}/\eps^2)$ oracle evaluation algorithm of Chakrabarty \etal~(STOC 2017) and the $\tilde{O}(n^{3/2}/\eps^2)$ oracle evaluation algorithm of Hamoudi \etal.

Further, we leverage a generalization of this result to obtain efficient algorithms for minimizing a broad class of nonconvex functions. For any function $f$ with domain $[0, 1]^n$ that satisfies $\frac{\partial^2f}{\partial x_i \partial x_j} \le 0$ for all $i \neq j$ and is $L$-Lipschitz with respect to the $L^\infty$-norm we give an algorithm that computes an $\eps$-additive approximate minimizer with $\tilde{O}(n \cdot \mathrm{poly}(L/\eps))$ function evaluation with high probability.

\end{titlepage}

\newpage

\section{Introduction}

A function $f$ which assigns real values to subsets of a finite universe $U$ is \emph{submodular} if it satisfies the \emph{decreasing marginal returns} property, i.e. $f(T \cup \{i\})-f(T) \le f(S \cup \{i\})-f(S)$ for all subsets $S \subseteq T \subseteq U$ and elements $i \not\in T$. Such functions are natural, arise in many applications, and have been studied extensively since the 1950s \cite{Cho55, Edmonds70, McC05, Fuji05, Lov83}. For example, the sizes of cuts in directed graphs or hypergraphs, the rank function of a matroid, and the entropy of subsets of random variables are all submodular. Further the utility functions of agents purchasing a subset of items is often assumed to be submodular. Given their prevalence, the optimization of submodular functions is fundamental to combinatorial optimization and both submodular function maximization \cite{Feige11, Chek14} and minimization have been studied extensively.

In this work, we focus on submodular function minimization (SFM), i.e. finding a subset $S \subseteq U$ minimizing $f(S)$. As submodular functions need not be monotone  and SFM generalizes multiple fundamental combinatorial optimization problems, including computing $s$-$t$ minimum cuts in directed graphs and hypergraphs, SFM is nontrivial. More recently, SFM has been applied to many problem domains, such as image segmentation \cite{B99, Koh08, Koh10}, speech analysis \cite{Lin10, Lin11, LinB11}, and machine learning \cite{Bach13, KG11}.

In this paper we consider the standard and well-studied model for SFM where $f$ can be accessed only through an \emph{evaluation oracle} which when queried with $S \subseteq U$ returns $f(S)$. For simplicity, we measure the complexity of our algorithms by the number of queries, i.e. \emph{oracle calls}, or \emph{function calls}, that we make to the evaluation oracle; the additional runtime of all new algorithms in this paper can be nearly linear in the number of oracle calls. Throughout the introduction, we refer to the time needed for an oracle call as $\eo$. An amazing result is that SFM can be solved with a number of queries polynomial in $n$, the number of elements in the universe $J$. This was demonstrated initially via the ellipsoid algorithm \cite{GLS84} spawning a long line of work faster algorithms. 

Previous research on algorithms for SFM has focused on three main regimes: strongly polynomial, weakly polynomial, and pseudopolynomial time \cite{Wolfe76, Fujis80, CJK14, LJJ15}. Letting $M$ be the maximum absolute value of the integer-valued submodular function $f$ on an $n$-element universe, strongly polymomial, weakly polynomial, and pseudopolynomial refer to algorithms whose runtimes are all polynomial in $n$ and  independent of $M$, logarithmic in $M$, and polynomial in $M$ respectively. For these regimes, the best known dependence on $n$ in terms of the number of oracle calls needed has a clear picture: nearly cubic in the strongly polynomial regime \cite{LSW15}, quadratic in the weakly polynomial regime \cite{LSW15}, and linear in the pseudopolynomial regime \cite{CLSW17}. 

We can also view these results in a slightly different way. Instead, let $f$ be a \emph{real-valued} submodular function with range $[-1,1]$, and consider the goal of finding an $\eps$-additive approximate minimizer. Here, it is natural to study approximate SFM algorithms whose runtimes are independent of $\eps$, depends logarithmically on $\eps$, and depends polynomially on $\eps$. These correspond to the strongly polynomial, weakly polynomial, and approximate regimes respectively\cite{CLSW17}. The best known runtime in the strongly polynomial and weakly polynomial regimes continue to be cubic and quadratic respectively in this view. However, despite the well-studied nature of SFM and clear picture in terms of $M$ dependencies, the runtime of approximate SFM is less understood. 
The state-of-the-art such runtime is $\O(n^{3/2}/\eps^2)$\footnote{Throughout, we use $\O$ to hide $\poly(\log n, \log(1/\eps), \log M)$ factors.} which was achieved by the contemporary work of \cite{HRRS19} and improved upon the $\O(n^{5/3}/\eps^2)$ runtime algorithm of \cite{CLSW17}.
In this paper, we close this gap and give a nearly linear time, $\O(n/\epsilon^2)$ time algorithm for $\eps$-approximate SFM.
. 

\paragraph{Nonconvex optimization:}
Another key motivation for the results of this paper is obtaining provably faster algorithms for obtaining global minimizers of broad classes of non-convex functions. Consider the problem $\min_{x \in \mathcal{X}} f(x)$ for a function $f : \R^n \to \R.$ For \emph{convex optimization}, when $f$ and $\mathcal{X}$ are both convex, there are numerous methods for solving this problem: gradient descent, cubic regularized newton, cutting plane, etc. On the other hand, the situation for nonconvex optimization, i.e. finding the global minimum of a nonconvex function $f$ in general, is computationally intractable: finding an $\eps$-approximate minimizer for a $k$-times continuously differentiable $f: \R^n \to \R$ requires $\Omega((1/\eps)^{n/k})$ evaluations of the function and its first $k$ derivatives, ignoring problem dependent parameters such as the Lipschitz smoothness of $f$, etc. \cite{NY83}. 

Nevertheless, as many practical problems, e.g. training neural networks and matrix completion, are nonconvex, it still important to understand what guarantees we can achieve for nonconvex optimization. Because computing an $\eps$-approximate global minimizer of a general convex function, as discussed, intractable in general, some work has focused on finding $\eps$-approximate stationary points or approximate local minima \cite{NP06, N12, CDHS16, BGMST17}. In addition, there are specific problems such as matrix completion where all local minima are in fact global minima \cite{GLM16}. Given these results, it would be tantalizing to find large classes of nonconvex functions for which we can find a global minimizer; however, this has been a challenging task achieved in only a few situations \cite{GG17, NGGD18, HSS19}.

In recent work, Bach \cite{Bach19} considered a class of nonconvex functions $f:[0,1]^n \to \R$ satisfying $\frac{\partial^2f}{\partial x_i \partial x_j} \le 0$ for all $i \neq j$, which are a continuous generalization of submodular functions. 
Several interesting functions satisfy this property, e.g. $f(x) = x^TQx$ where $Q$ is symmetric with negative off diagonal entries and $f(x) = g(\sum_i c_ix_i)$ for some concave function $g: \R \to \R$ and positive weights $c_i$. The former function is neither convex nor concave, and the latter is concave.

Despite the fact that these functions can be nonconvex, \cite{Bach19} provided an algorithm to find $\eps$-approximate global minimizers in time polynomial in $n, \eps$, and problem dependent parameters.  In our work, we improve upon Bach's cubic dependence on $n$ and show that these functions can in fact be minimized \emph{almost as efficient as convex functions} in terms of the best known methods: nearly linear in the dimension $n$, and polynomial in $\eps$ and the $L^\infty$ Lipschitz constant.

\subsection{Our results}

In this paper we address a key open problem in the work of \cite{CLSW17} and \cite{HRRS19} regarding whether we can achieve a nearly linear runtime for approximate SFM. We resolve this problem in this paper, giving an $\O(n\eps^{-2} \cdot \eo)$ time algorithm for approximate SFM. This also directly improves the previous pseudopolynomial time algorithms in terms of their dependence on $M$, the range of an integer submodular function. Further, due to the subgradient oracle lower bound given in \cite{CLSW17} and the fact that a subgradient oracle yields more information than an evaluation oracle, this bound is known to be optimal up to the dependence on $\eps$.

\begin{theorem}[Nearly linear time submodular function minimization]
\label{thm:main}
Given a submodular function $f: \{0, 1\}^n \to [-1, 1]$ and an $\eps > 0$, we can compute a random set $S$ with \[ \E[f(S)] \le \min_{T\subseteq [n]} f(T) + \eps \] in $\O(n/\eps^2)$ calls to an evaluation oracle for $f$.
\end{theorem}
We can convert the guarantee of \cref{thm:main} to a w.h.p.\footnote{Throughout our results, we use w.h.p. to mean ``with high probability in $n$".} guarantee as follows. Note that the probability that $f(S) > \min_{T\subseteq [n]} f(T) + 2\eps$ is at most $1/2$ by Markov's inequality. Consequently, we can amplify to the success probability to $1 - \frac{1}{\poly(n)}$ by running the algorithm $O(\log n)$ times to half the expected error and outputing the smallest value.

We also achieve \emph{sublinear} time algorithms for pseudopolynomial SFM in settings where we know that the minimizer of $f$ is $s$-sparse, i.e. only has $s$ nonzero entries.

\begin{theorem}[Pseudopolynomial submodular function minimization]
\label{thm:sparse}
Consider an integer valued submodular function $f: \{0, 1\}^n \to [-M, M]$ with $s$-sparse minimizer, i.e. there is a set $S^\opt \in \argmin_{S \subseteq \{0, 1\}^n} f(S)$ satisfying $|S^\opt| \le s.$ Then we can compute an exact minimizer of $f$ in $\O(sM^2)$ calls to an oracle for $f$ w.h.p.
\end{theorem}
Note that for small $M$ (say $M = \O(1)$) and $s = O(n^{1-\delta})$ for some $\delta > 0$, this algorithm uses a number of oracle calls to $f$ which is \emph{sublinear} in $n$. This is the first sublinear time algorithm for SFM. 
The previous bottleneck for obtaining such sublinear results was that it seemed necessary to compute a full subgradient of the Lovasz extension, a well-known continuous extension of submodular functions, which naively requires $\Omega(n)$ oracle calls. We overcome this by designing an algorithm that computes all $O(M)$ nonzero entries of the subgradient at $0$ with $\O(M^2)$ oracle calls.

These results makes progress towards completing the picture for SFM algorithms: strongly polynomial algorithms use a cubic number of queries, weakly polynomial algorithms use a quadratic number of queries, and pseudopolynomial/approximate algorithms make a linear number of queries. 

Following the work of Bach \cite{Bach19}, our results extend to a more general class of submodular functions not necessarily defined on $\{0, 1\}^n$, such as those defined on $[k]^n$ for a positive integer $k$. Leveraging this result, we obtain a nearly linear time algorithm for computing approximate minimizers of a class of nonconvex functions studied by Bach \cite{Bach19}, improving upon the cubic running time given in that paper. 

\begin{theorem}[Nearly linear time continuous submodular function minimization]
\label{thm:nonconvex}
Let $f: [0, 1]^n \to \R$ be a twice differentiable function with  $\frac{\partial^2 f(x)}{\partial x_i \partial x_j} \le 0$ for all $i \neq j$. There is an algorithm that computes an $\eps$-additive approximate minimizer of $f$ in $\O(nL^6 / \eps^6)$ function evaluation calls w.h.p., where $L$ is the $L^\infty$-Lipschitz constant of $f$.
\end{theorem}

\subsection{Previous Work}
The first polynomial time algorithm for SFM was via the ellipsoid algorithm \cite{GLS84}. This spawned a line of work on faster algorithms (see \cref{fig:results} for the state-of-the-art bounds) and combinatorial algorithms \cite{Cun85, IFF00, Sch00, IO09}, which were achieved later. Building on a long line of work on SFM, Lee \etal~\cite{LSW15} gave the current state-of-the-art running times for weekly polynomial SFM, $O(n^2 \log nM \cdot \eo + n^3 \log^{O(1)}nM)$, and strongly polynomial SFM, $O(n^3 \log^2 n \cdot \eo + n^4 \log^{O(1)} n)$. See \cite{LSW15} for more comprehensive coverage of previous improvements.

Additionally, there has been work towards understanding \emph{pseudopolyomial} algorithms for SFM. Specifically, the Fujishige-Wolfe \cite{Wolfe76, Fujis80} algorithm which is often used in practice can be shown to run in pseudopolynomial time $O(n^2M^2 \cdot \eo + n^3M^2)$ \cite{CJK14, LJJ15}. More recently, Chakrabarty \etal~\cite{CLSW17} gave a nearly linear pseudopolynomial algorithms for SFM with runtime $\O(nM^3 \cdot \eo)$. Additionally, they studied the problem of approximate SFM, that is minimizing a real-valued submodular function $f$ with range $[-1, 1]$ to additive $\eps$ error. They achieved a subquadratic $\O(n^{5/3} \eps^{-2} \cdot \eo)$ time algorithm for this problem. They also studied SFM in the case where $f$ is known to have a $s$-sparse minimizer, i.e. the minimizer of $f$ has only $s$ nonzero entries, achieving an $\O \left ( (n + s n ^{2/3}) \eo \epsilon^{-2} \right )$ algorithm.

Simultaneously with this work, Hamoudi \etal~\cite{HRRS19} improved the runtime of approximate SFM to $\O(n^{3/2}\eps^{-2} \cdot \eo)$. Additionally, they also achieved a $\O(n^{5/4}\eps^{-5/2} \cdot \eo)$ quantum algorithm for approximate SFM through a new method for sampling with high probability
$T$ independent elements from any discrete probability distribution of support size $n$ in time $O(\sqrt{Tn}).$

\begin{table}[h]
\centering{}%
\begin{tabular}{|c||c|c|c|}
\hline 
Regime & {\small{}Previous Best Running Time} & {\small{}Our Result}\tabularnewline
\hline 
\hline 
{\small{}Strongly Polynomial} & \multicolumn{2}{c|}{{\small{}$O(n^{3}\log^{2}n\cdot\eo+n^{4}\log^{O(1)}n)$ \cite{LSW15}}}\tabularnewline
\hline 
{\small{}Weakly Polynomial} & \multicolumn{2}{c|}{{\small{}$O(n^{2}\log nM\cdot\eo+n^{3}\log^{O(1)}nM)$\cite{LSW15}}}\tabularnewline
\hline 
{\small{}Pseudopolynomial} & {\small{}$\O(nM^{3}\cdot\eo)$\cite{CLSW17}} & \textbf{\small{}$\O(nM^2\cdot\eo)$}\tabularnewline
\hline 
\multirow{1}{*}{{\small{}$\eps$-Approximate}} & \multirow{1}{*}{{\small{}$\O(n^{3/2}\cdot\eo/\eps^2)$\cite{HRRS19}}} & \textbf{\small{}$\tilde{O}(n\cdot\eo/\eps^2)$} \tabularnewline
\hline
\multirow{1}{*}{{\small{}Sparse Pseudopolynomial}} & \multirow{1}{*}{{\small{}$\O((n+sM^3)\cdot\eo)$\cite{CLSW17}}} & \textbf{\small{}$\O(sM^2 \cdot \eo)$} \tabularnewline
\hline
\end{tabular}\caption{\label{fig:results} {\small{}Running times for minimizing a submodular function $f$ on subsets of an $n$ element set. $\eo$ denotes the time needed to make an oracle call to $f$. In all but the approximate SFM regime, $f$ is integer valued with maximum absolute value $M$. In the approximate SFM regime, $f$ is real valued with range $[-1, 1]$. $s$ is the sparsity of the minimizer of $f$. Table adapted from \cite{CLSW17}.}}
\end{table}

Additionally, there has been work towards extending the notion of submodularity beyond functions defined on subsets of a universe $U$. Bach \cite{Bach19} has shown that the notion of submodularity extends naturally to functions defined on $[k]^n$ (instead of $\{0,1\}^n$) and even to functions defined on continuous domains such as $[0, 1]^n.$ This work shows how to extend the classical polynomial time algorithms for submodular optimization to this setting, and gives polynomial time algorithms for optimizing a large class of nonconvex functions.

\subsection{Organization}
For the remainder of the introduction, we give an overview for our techniques in \cref{sec:overview}.
In \cref{sec:prelim}, we state the necessary preliminaries for our algorithms.
In \cref{sec:sfm2} we give our main algorithm for nearly linear time SFM and prove \cref{thm:sparse}.
In \cref{sec:sparse} we give our sublinear time algorithms for pseudopolynomial SFM when the minimizer is sparse and prove \cref{thm:sparse}.
Finally, in \cref{sec:sfmk} we extend our earlier results to submodular functions on the domain $[k]^n$ and $[0, 1]^n$, and prove \cref{thm:nonconvex} in \cref{sec:nonconvex}.



\subsection{Overview}
\label{sec:overview}
Here we give a less technical overview of the ideas behind our algorithm. For simplicity, we only describe our algorithm in the situation when $f$ is a standard submodular function on $\{0,1\}^n$. For more technical discussion and discussion about the sparse regime and submodular functions over $[k]^n$, see \cref{sec:sfm2}, \cref{sec:sparse}, and \cref{sec:sfmk}.

Our algorithms, like those of \cite{CLSW17}, are based on projected stochastic subgradient descent on the Lovasz extension of $f$, which is a well-known continuous convex extension of $f$ (see \cref{def:lovasz}). The algorithms of \cite{CLSW17} exploited submodularity to build a data structure and get gradient updates with fewer evaluations than the $O(n)$ required by the na\"ive method. 
To obtain our result we leverage the techniques built by \cite{CLSW17} but show that a more efficient binary tree based data structure can be built to support gradient estimates with little preprocessing. 

More precisely, the algorithm of \cite{CLSW17} computes $x_0, x_1, \cdots, x_T$, a (stochastic) sequence of points, where each $x_{i+1}$ is computed by taking a stochastic subgradient step from $x_i$.
To do this, the algorithm writes the gradient at $x_{i+1}$ (we'll denote it as $g(x_{i+1})$) as the sum of smaller terms of the form $g(x_a) - g(x_b)$ and $g(x_0)$, estimates each, and sums them. As the number of terms summed increases, the variance of the estimate grows and the convergence rate of subgradient descent decreases. To achieve there fastest algorithm \cite{CLSW17} thereby trades off leveraging such stochastic estimates and recomputing the initial estimator.

In this paper we improve this datastructure by, as we step through the trajectory $x_1,\ldots, x_T$, choosing to evaluate $g(x_a) - g(x_b)$ at carefully chosen intervals along the trajectory. This allows us to amortize the maintenance of the data structure while simultaneously maintaining a low variance of the resulting stochastic gradient.   
This leads to our nearly linear time algorithm. 
We hope that this general framework of using data structures to maintain the ability to do point updates and sample gradients can find uses in other optimization methods where we desire sublinear gradient calls, such as coordinate descent.

In order to extend our results to the domain $[k]^n$, we use the continuous extension of a submodular function $f$ developed by Bach \cite{Bach19}, which is the analogue of the Lovasz extension. We show that our algorithms extend to this setting.

Finally, we explain our key ingredient to obtaining sublinear time algorithms in the regime where $f$ is integer valued with maximum absolute value $M$ and has a sparse minimizer. The main idea behind the algorithm is to compute an initial subgradient at $0$ that doesn't require computing all $n$ coordinates of the subgradient, which na\"ively requires $n$ function calls. To do this, we use that the origin has many subgradients, and develop an algorithm that can find one such subgradient for which we can compute all its nonzero entries in $\O(M^2)$ function calls. After this, we can simply plug this initial subgradient into our earlier algorithms and get the desired result.

\section{Preliminaries}
\label{sec:prelim}

Here we provide notation and basic facts about classic submodular functions. Preliminaries for submodular functions on $[k]^n$ and continuous submodular functions are deferred to \cref{sec:sfmk}.

\paragraph{Miscellaneous notation.} We let $[n] \defeq \{1, 2, \dots, n\}$. For $a, b \in \R$ we let $[a, b] \defeq \{x : a \le x \le b\}.$ For permutation $P = \{P_1, P_2, \cdots, P_n\}$ of $[n]$, we let $P[j] \defeq \{P_1, P_2, \cdots, P_j\}$ be the set containing the first $j$ elements of $P$. For a point $x \in \R^n$ we call a permutation $P$ of $[n]$ \emph{consistent} with $x$ if $x_{P_1} \ge x_{P_2} \ge \cdots \ge x_{P_n}.$ We let $e_1, e_2, \cdots, e_n$ denote the standard basis vectors for $\R^n$, so that $e_i$ is the vector with a $1$ in the $i$-th coordinate and $0$ in all other coordinates. We call a vector $s$-sparse if it has at most $s$ nonzero entries.

\paragraph{Submodular functions.} Let $\{0,1\}^n \subseteq \R^n$ denote the set of $n$-tuples, where each coordinate is either $0$ or $1$. There is a natural bijection between $x \in \{0,1\}^n$ and subsets $S \subseteq [n]$ where $x_i$ being $1$ corresponds to element $i$ being in the set. We use these interchangeably. Throughout, we let $f: \{0, 1\}^U \to \R$ be the submodular function we are trying to optimize, where $U$ is a ground set. Without loss of generality, we assume $U = [n]$. Additionally, we assume that $f(\emptyset) = 0$, which we can enforce by subtracting a constant from all values of $f$ while preserving submodularity. We say that a function $f: \{0, 1\}^n \to \R$ is submodular if it satisfies the property of decreasing marginal returns, specifically for all sets $S \subseteq T \subseteq [n]$ and element $i \not\in T$, we have 
\[ f(S \cup \{i\})-f(S) \ge f(T\cup \{i\})-f(T). \] An alternate (but equivalent) definition is that for all $S, T \subseteq [n]$ we have that
\[ f(S)+f(T) \ge f(S\cap T)+f(S\cup T). \]
In this work, we measure the complexity of our algorithms through the number of calls we make to an evaluation oracle for $f$, as the additional runtime of all new algorithms in this paper can be nearly linear in the number of oracle calls.

\paragraph{Lovasz extension.}
 The Lovasz extension is a well-known continuous, convex extension of a submodular function $f: \{0, 1\}^n \to \R$ to a function $\hf: [0, 1]^n \to \R.$
We now state its definition.

\begin{definition}[Lovasz extension] 
\label{def:lovasz}
Given a submodular function $f:\{0,1\}^n \to \R$, the \emph{Lovasz extension of $f$}, denoted as $\hf:[0, 1]^n \to \R$, is defined for any $x \in [0, 1]^n$ as
\begin{equation} \label{eq:lov} \hf(x) = \sum_{j=1}^n (f(P[j])-f(P[j-1]))x_{P_j}, \end{equation}
where $P = \{P_1, P_2, \cdots, P_n\}$ is a permutation which is consistent with $x$.
\end{definition}

We leverage the following well known properties of the Lovasz extension \cite{Lov83, Fuji05}.
\begin{theorem}
\label{thm:lovasz}
Let $f: \{0, 1\}^n \to \R$ be a submodular function, and let $\hf$ be its Lovasz extension. We have that: $\hf$ is convex; for all $x \in \{0, 1\}^n$, $\hf(x) = f(x)$; and $\min_{x \in [0, 1]^n} \hf(x) = \min_{S \subseteq [n]} f(S).$
Additionally, the vector $g(x) \in \R^n$ defined by $g(x)_{P_j} \defeq f(P[j])-f(P[j-1])$ for $1 \le j \le n$ is a subgradient of $\hf$ at x, where $P = (P_1, P_2, \cdots, P_n)$ is any permutation consistent with $x$.
\end{theorem}
Note that the vector $g(x)$ as defined in \cref{thm:lovasz}, despite being a subgradient of $\hf$ of $x$, only depends on $P$. Thus, sometimes we define the \emph{gradient (at zero) associated with permutation $P$}, denoted $g^P$, as $g^P_{P_j} \defeq f(P[j])-f(P[j-1])$ for $1 \le j \le n.$ 

We now explain (and this is standard) that given a point $x \in [0, 1]^n$, we can find a set $S \subseteq [n]$ with $f(S) \le \hf(x)$ in $O(n)$ oracle calls to $f$. In other words, we only need to pay an extra $\O(n)$ oracle calls to convert an approximate minimizer of the Lovasz extension $\hf$ of $f$ to an approximate minimizer of $f$ itself. For completeness, we state this as a lemma and prove it below.

\begin{lemma}[Going from $\hf$ to $f$]
\label{lemma:discretize}
For a point $x \in [0,1]^n$, we can in $O(n)$ oracle calls compute a set $S$ such that $f(S) \le \hf(x).$ In particular, we can go from an $\eps$-additive approximate minimizer of $\hf$ to an $\eps$-additive approximate minimizer of $f$ in $O(n)$ oracle calls.
\end{lemma}
\begin{proof}
We can rewrite \cref{eq:lov} as
\[ \hf(x) = f(P[n])x_{P_n} + \sum_{j=1}^{n-1} f(P[j])(x_{P_j}-x_{P_{j+1}}), \] thus $\hf(x)$ is a non-negative linear combination of $f(\emptyset), f(P[1]), f(P[2]), \cdots, f(P[n])$. Therefore, either $f(\emptyset) \le \hf(x)$ or there is an $1 \le i \le n$ with $f(P[i]) \le \hf(x)$, as desired. 
\end{proof}

\paragraph{Subgradient descent.} Our algorithms are primarily based on (projected stochastic) subgradient descent. For a convex function $f$ on a convex compact set $S \subseteq \R^n$, we say that a vector $g$ is a \emph{subgradient} of $f$ at $x \in S$ if for all $y \in S$ we have that
\[ f(y)-f(x) \ge g^T (y-x). \] We let $\partial f(x)$ be the set of all subgradients of $f$ at $x$. A \emph{subgradient oracle} for $f$ is an algorithm which at a point $x \in S$ returns a vector $g$ with $g \in \partial f(x).$ A \emph{stochastic subgradient oracle} for $f$ is an algorithm which at a point $x \in S$ returns a stochastic vector $\bg(x)$ with $\E[\bg(x)] \in \partial f(x)$.\footnote{Throughout, we use boldface (e.g. $\bg$) for stochastic variables and normal text (e.g. $g$) for not stochastic variables}  Finally, intermediate points computed during projected subgradient descent may lie outside $S$. We define the \emph{projection} of a point $y$ onto $S$ to be \[ \proj(y, S) = \argmin_{x \in S} \| x - y \|_2^2. \] We now state a theorem which contains the guarantees of projected stochastic subgradient descent which we use. The version we state is adapted from \cite{Bubeck15} and suffices for our purposes.

\begin{theorem}[Projected stochastic subgradient descent \cite{Bubeck15}]
\label{thm:grad}
Let $f$ be a convex function on a compact convex set $S \subseteq \R^n$ and $\bg$ be a stochastic subgradient oracle for $f$. Define parameters $R^2 \defeq \max_{x \in S} \frac{1}{2}\|x\|_2^2$, $B$ such that $\E[\|\bg(x)\|_2^2] \le B^2$ for all $x \in S$ and consider the following iterative algorithm
\[ x_1 \defeq \argmin_{x \in S} \|x\|_2^2 \] 
and \[ x_{i+1} = \proj(x_i - \eta \bg(x_i), S) \text{ for } i \in [T - 1]  \] Then for $\eta = \frac{R}{B}\sqrt{\frac{2}{T}}$, we have that \[ \E\left[f\left(\frac{1}{T}\sum_{i=1}^T x_i\right)\right] \le \min_{x \in S} f(x) + RB\sqrt{\frac{2}{T}}. \]
\end{theorem}
Note that for $T = 2R^2 B^2/\eps^2$ in \cref{thm:grad} we achieve additive error $\eps$ off the minimum function value in expectation.

\section{Submodular Function Minimization over $\{0, 1\}^n$}
\label{sec:sfm2}
In this section we present our improved algorithm for SFM over $\{0, 1\}^n$. For a submodular function $f:\{0,1\}^n \to \R$, our algorithms perform projected stochastic subgradient descent on the Lovasz extension $\hf$ of $f$.

Looking at the guarantees of \cref{thm:grad}, we want to design an algorithm that can compute stochastic subgradients with low expected $\ell_2$ norm without having to make many oracle calls to $f$. Specifically, in the case of \cref{thm:main}, we show how to construct an algorithm that
\begin{itemize}
\item Computes a sequence of points $x_1, x_2, \cdots, x_T \in [0, 1]^n$ and stochastic subgradients $\bg^i$ of $f$ at $x_i$, where $x_1 = 0$ and $x_{i+1} = \proj(x_i - \eta \bg^i, [0, 1]^n).$
\item Makes $\O(T)$ oracle calls to $f$.
\item Each stochastic subgradient $\bg^i$ is $1$-sparse.
\item Each stochastic subgradient $\bg^i$ has $\E[\| \bg^i \|_2^2] = \O(1).$
\end{itemize}
By the guarantees of \cref{thm:grad}, choosing $T = \O(n/\eps^2)$ suffices to prove \cref{thm:main}, as $R^2 = n$ and $B^2 = \O(1).$

\subsection{Subgradients of the Lovasz extension}
Here we state important results on the struture of the subgradients of the Lovasz extension. We use this structure in order to sample stochastic subgradients of $\hf$ in sublinear time. \cref{lemma:l1} is due to Jegelka and Blimes \cite{JB11} (also Hazan and Kale \cite{HK12}). All of \cref{lemma:l1}, \cref{lemma:change}, and \cref{lemma:sum} were proven in \cite{CLSW17}.

The first lemma is a bound on the $L^1$ norm of the subgradients.
\begin{lemma}
\label{lemma:l1}
For a submodular function $f: \{0, 1\}^n \to [-M, M]$, all subgradients $g$ of the Lovasz extension satisfy $\|g(x)\|_1 \le 3M.$
\end{lemma}

The second lemma allows us to relate the gradients of two points $x, y \in \{0, 1\}^n$ whose difference $x-y$ is a strictly positive (or negative) vector.
\begin{lemma}
\label{lemma:change}
Let $x \in [0, 1]^n$ and let $d \in \R_{\ge 0}^n$ be such that $y = x + d$ (resp. $y = x-d$). Let $S$ denote the non-zero coordinates of $d$. Then for all $i \not\in S$ we have $g(x)_i \ge g(y)_i$ (resp. $g(x)_i \le g(y)_i$).
\end{lemma}

The final lemma allows us to efficiently compute the sum of multiple contiguous coordinates of a subgradient.
\begin{lemma}
\label{lemma:sum}
Let $x \in [0, 1]^n$ and let $P$ be the permutation consistent with $x$. Then we have for any integers $1 \le a \le b \le n$ that \[ \sum_{i=a}^b g(x)_{P_i} = f(P[b]) - f(P[a-1]). \]
\end{lemma}

\subsection{Nearly linear time approximate submodular function minimization}

In this section we provide a nearly linear time algorithm for minimizing a submodular function $f: \{0,1\}^n \to [-1, 1]$ to additive error $\eps.$ We give a randomized algorithm that uses at most $\O(n/\eps^2)$ oracle calls to $f$ that computes a point $x \in [0, 1]^n$ with $\hf(x) \le \min_T f(T) + \eps$, where $\hf$ is the Lovasz extension of $f$.

Our algorithm follows the broad framework of \cite{CLSW17} which minimizes 
the Lovasz extension using projected stochastic subgradient descent. 
By \cref{thm:grad}, this algorithm yields an $\eps$-additive approximate minimizer in $\O(n/\eps^2)$ provided each subgradient has expected $\O(1)$ $\ell_2$ norm.
Because na\"ively computing a full subradient gradient $g$ of $\hf$ at $x$ requires $\Omega(n)$ oracle calls, to achieve our runtime improvements we must do something more sophisticated to compute stochastic subgradients. To overcome this issue, as in \cite{CLSW17}, we leverage \cref{lemma:l1}. This lemma implies that there is stochastic gradient oracle which outputs subgradients which are both sparse and have low $\ell_2$ norm. Indeed, because $\|g\|_1 \le 3$ (by \cref{lemma:l1}) for all subgradients $g$ of the Lovasz extension, we can compute a $1$-sparse stochastic subgradient $\bg$ with $\E[\|\bg\|_2^2] \le \|g\|_1^2 \le 9$: sample $\bg = \mathrm{sign}(g_i)\|g\|_1 e_i$ with probability $|g_i|/\|g\|_1$.

While  \cref{lemma:l1} does give a sparse sparse stochastic subgradient oracle with low $\ell_2$ norm, a na\"ive implementation would require knowing all of $g$ and therefore naively, $\Omega(n)$ oracle calls. To get around this issue, a key insight of \cite{CLSW17} was that if $g$ was guaranteed to have all positive coordinates, we could use a binary search to sample a stochastic subgradient with $O(1)$ variance in $\O(1)$ oracle calls by applying  \cref{lemma:sum} to sample recursively. This procedure simply samples an interval with probability proportional to the sum of its coordinates, computing the sum using \cref{lemma:sum}
(further details are given in the proof of \cref{lemma:sample} in \cref{sec:proofs}). Unfortunately, this  only works in the case where all coordinates of $g$ are positive, as then there is no cancellation when we compute the sum of coordinates in an interval.

However, imagine that we have already computed the gradient $g^0$ at $x_0$ (the starting point of our method) and $x_1 = x_0 - \eta\bg^0$ where $\bg^0$ is a $1$-sparse stochastic subgradient at $x_0$. To sample a stochastic subgradient $\bg^1$ at $x_1$ we write $g^1 = g^0 + (g^1 - g^0).$ In order to sample $\bg^1$, we sample an estimate of $g^0$, an estimate of $g^1 - g^0$, and sum them. Call this estimate $\bd$. If we could efficiently sample a $1$-sparse $O(1)$ variance estimate of $g^1-g^0$, then the resulting estimate $\bd$ would be $2$-sparse with $O(1)$ variance. To get $\bg^1$ simply sample twice a random nonzero coordinate of $\bd$. Thus $\bg^1$ would be $1$-sparse with $O(1)$ variance. 

To efficiently sample an estimate of $g^1 - g^0$, \cite{CLSW17} noted that if the difference $x_1-x_0$ is $1$-sparse, then by submodularity one can show that $g^1-g^0$ can be split into $O(1)$ intervals, each of which is either all positive or all negative. We can then sample this efficiently by the same algorithm for sampling all positive gradients above.
This intuition is formalized and generalized in the following lemma which is a slight modification of Lemma 12 proven in \cite{CLSW17}.
It gives us the ability to efficiently sample a sparse, low variance estimate of $g(x) - g(y)$ where $x-y$ is sparse.
For example, in the paragraph above where $x_1 - x_0$ is $1$-sparse, we could efficiently sample a $1$-sparse estimate of $g^1 - g^0.$

\begin{restatable}{lemma}{restatesample}
\label{lemma:sample}
Let $f:\{0, 1\}^n \to [-1, 1]$ be a submodular function with Lovasz extension $\hf$. Let $g$ denote the subgradients of $\hf$. Let $x, y \in [0, 1]^n$ be vectors such that $y-x$ is $k$-sparse. There is a data structure which after $O(k)$ calls to $f$ of preprocessing, supports the following: sample a $1$-sparse random variable $\bz$ with $\E[\bz] = g(y)-g(x)$ and $\E[\|\bz\|_2^2] = O(1)$ in $\O(1)$ calls to $f$. Preprocessing is called through $\Process(x, y, f)$, and the sampling is called through $\Sample(x, y, f).$
\end{restatable}

In other words, \cref{lemma:sample} implies for fixed $x,y$, we build a data structure with $O(k)$ oracle calls that supports sampling estimates of $g(y) - g(x)$ using $\O(1)$ oracle calls per sample We give the proof in \cref{sec:proofs} for completeness. Careful application of this lemma and the idea of sampling from gradient differences yields the runtimes in \cite{CLSW17} and, with modification to the datastructure, \cite{HRRS19}.

Where we depart from  \cite{CLSW17} (and improve upon it) is how we use the data structure of \cref{lemma:sample}. 
Recall that at iteration $t$, we want to sample an estimate of $g(x_t)$. Instead of using $g(x_t) - g(x_0)$ as above, we will carefully choose a short sequence, $x_{i_0}, x_{i_1}, \dots, x_{i_m}$, where $i_0 = 0$ and $i_m = t$ for $m = \O(1)$.
Now, we sample an estimate of $g(x_t)$ using the identity $g(x_t) = g(x_0) + \sum_{j=0}^{m-1}\left[g(x_{i_{j+1}})-g(x_{i_j})\right]$.
Specifically, we sample an estimate of $g(x_0)$ and each of the remaining terms $g(x_{i_{j+1}})-g(x_{i_j})$ using \cref{lemma:sample}, and sum the estimates.
Note that the sum is $\O(1)$-sparse and has $\O(1)$-variance by \cref{lemma:sample}. Now, we can sample a $1$-sparse estimate of this sum with $\O(1)$ variance.
The key to our algorithm is then choosing the sequence so that the amortized cost of preprocessing all segments $(x_{i_j}, x_{i_{j+1}})$ is $\O(1)$ per iteration.

It suffices to choose the segments using the binary representation of the step counter $t$. 
For example, for $11$ ($1011$ in binary) we would choose $0,8,10,11$ ($0, 1000,1010,1011$ in binary). See \cref{fig:tree} for the corresponding segments.

\begin{figure}
\caption{Intervals processed in line \ref{line:process} and a decomposition of $[0, 11]$}
\begin{center}
\begin{tikzpicture}
  \draw[blue, line width = 0.5mm] (0, 30)--(11, 30);
  \draw[blue, line width = 0.5mm] (0, 29.8)--(0, 30.2);
  \draw[blue, line width = 0.5mm] (11, 29.8)--(11, 30.2);
  \draw node at (0, 30)[above=0.2]{$0$};
  \draw node at (11, 30)[above=0.2]{$11$};
  \draw[red, line width = 0.5mm] (0, 29.5)--(8, 29.5);
  \draw[red, line width = 0.5mm] (0, 29.3)--(0, 29.7);
  \draw[red, line width = 0.5mm] (8, 29.3)--(8, 29.7);
  \draw (0, 29)--(4, 29);
  \draw (8, 29)--(12, 29);
  \draw (0, 28.8)--(0, 29.2);
  \draw (4, 28.8)--(4, 29.2);
  \draw (8, 28.8)--(8, 29.2);
  \draw (12, 28.8)--(12, 29.2);
	\draw (0, 28.5)--(2, 28.5);
	\draw (4, 28.5)--(6, 28.5);
	\draw[red, line width = 0.5mm] (8, 28.5)--(10, 28.5);
	\draw (12, 28.5)--(14, 28.5);
	\draw (0, 28.3)--(0, 28.7);
	\draw (2, 28.3)--(2, 28.7);
	\draw (4, 28.3)--(4, 28.7);
	\draw (6, 28.3)--(6, 28.7);
	\draw[red, line width = 0.5mm] (8, 28.3)--(8, 28.7);
	\draw[red, line width = 0.5mm] (10, 28.3)--(10, 28.7);
	\draw (12, 28.3)--(12, 28.7);
	\draw (14, 28.3)--(14, 28.7);
	\draw (0, 28)--(1, 28);
	\draw (2, 28)--(3, 28);
	\draw (4, 28)--(5, 28);
	\draw (6, 28)--(7, 28);
	\draw (8, 28)--(9, 28);
	\draw[red, line width = 0.5mm] (10, 28)--(11, 28);
	\draw (12, 28)--(13, 28);
	\draw (14, 28)--(15, 28);
	\draw (0, 27.8)--(0, 28.2);
	\draw (1, 27.8)--(1, 28.2);
	\draw (2, 27.8)--(2, 28.2);
	\draw (3, 27.8)--(3, 28.2);
	\draw (4, 27.8)--(4, 28.2);;
	\draw (5, 27.8)--(5, 28.2);
	\draw (6, 27.8)--(6, 28.2);
	\draw (7, 27.8)--(7, 28.2);
	\draw (8, 27.8)--(8, 28.2);
	\draw (9, 27.8)--(9, 28.2);
	\draw[red, line width = 0.5mm] (10, 27.8)--(10, 28.2);
	\draw[red, line width = 0.5mm] (11, 27.8)--(11, 28.2);
	\draw (12, 27.8)--(12, 28.2);
	\draw (13, 27.8)--(13, 28.2);
	\draw (14, 27.8)--(14, 28.2);
	\draw (15, 27.8)--(15, 28.2);
	\draw node at (0, 28)[below=0.2]{$0$};
	\draw node at (1, 28)[below=0.2]{$1$};
	\draw node at (2, 28)[below=0.2]{$2$};
	\draw node at (3, 28)[below=0.2]{$3$};
	\draw node at (4, 28)[below=0.2]{$4$};
	\draw node at (5, 28)[below=0.2]{$5$};
	\draw node at (6, 28)[below=0.2]{$6$};
	\draw node at (7, 28)[below=0.2]{$7$};
	\draw node at (8, 28)[below=0.2]{$8$};
	\draw node at (9, 28)[below=0.2]{$9$};
	\draw node at (10, 28)[below=0.2]{$10$};
	\draw node at (11, 28)[below=0.2]{$11$};
	\draw node at (12, 28)[below=0.2]{$12$};
	\draw node at (13, 28)[below=0.2]{$13$};
	\draw node at (14, 28)[below=0.2]{$14$};
	\draw node at (15, 28)[below=0.2]{$15$};
\end{tikzpicture}
\end{center}
\label{fig:tree}
\end{figure}
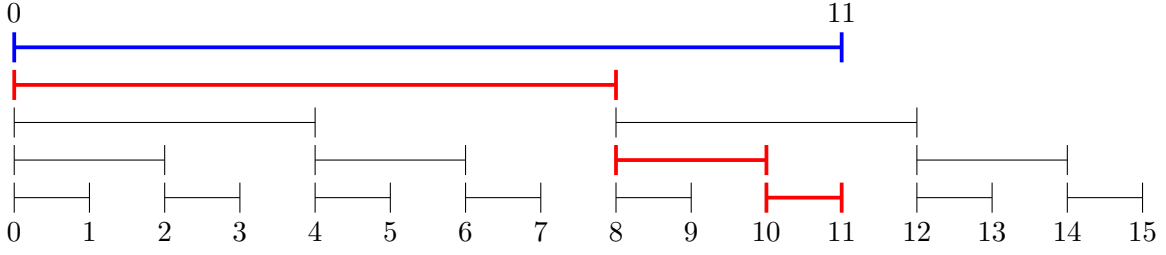

At this point, we are ready to state our algorithm.
\begin{algorithm}[h!]
\caption{$\SFM(f, \eps)$. Takes a submodular function $f: \{0,1\}^n \to [-1, 1]$ and returns a random point $x \in [0, 1]^n$ with $\E[\hf(x)] \le \min_T f(T) + \eps$.}
\begin{algorithmic}[1]
\State $T \assign \O(n/\eps^2).$ \label{line:init1}
\State $x_0 \assign 0 \in \R^n.$ \label{line:init2}
\State $g^0 \assign g(x_0).$ \label{line:init3}
\For{$i=1$ to $T$}\label{line:forloop}
	\State $b \assign $ the number of $1$ bits in the binary representation of $i-1$. \label{line:startbinary}
	\State $k_0 \assign i-1$, $k_{j+1} \assign k_j - 2^{\nu_2(k_j)}$ for $0 \le j \le b-1$. \Comment{$\nu_2(y)$ is the maximum integer $t$ such that $2^t$ divides $y$ } \label{line:makesegments} 
	\State $\bd \assign \|g^0\|_1 \cdot \mathrm{sign}(g^0_k)e_k$ with probability $|g^0_k|/\|g^0\|.$ \Comment{Estimate of $g^0$} \label{line:getd} 
	\State $\bg^{\star i} \assign \bd + \sum_{j=0}^{b-1} \Sample(x_{k_{j+1}}, x_{k_j}, f).$ \label{line:getgstar}
	\State Let $c_1, c_2, \cdots, c_s$ be the nonzero coordinates of $\bg^{\star i}$. \Comment{$s \le b+1$} \label{line:sparse1}
	\State $\bg^i \assign s \cdot \bg^{ \star i}_{c_k} \cdot e_{c_k}$ with probability $\frac{1}{s}.$ \Comment{$\bg^i$ is $1$-sparse} \label{line:sparse2}
	\State $x_i \assign \proj(x_{i-1} - \eta \bg^i, [0, 1]^n).$ \label{line:endbinary}
	\State $\Process(x_{i-2^{\nu_2(i)}}, x_i, f).$ \label{line:process}
\EndFor
\State \Return $\frac{1}{T+1} \sum_{i=0}^T x_i.$
\label{algo:sfm}
\end{algorithmic}
\end{algorithm}

\paragraph{Description of \cref{algo:sfm}.} Lines \ref{line:init1}, \ref{line:init2}, \ref{line:init3} initialize the starting point, number of iterations, and gradient $g^0$ at the the initial point. Line \ref{line:forloop} corresponds to the projected stochastic gradient descent loop.
Lines \ref{line:startbinary}, \ref{line:makesegments} compute the segments using the binary representation of the iteration counter $i$.
Lines \ref{line:getd}, \ref{line:getgstar} use the precomputed data structures to sample a stochastic subgradient in $\O(1)$ time.
Lines \ref{line:sparse1}, \ref{line:sparse2} turn this stochastic subgradient into a $1-$sparse stochastic subgradient which is used in the descent step in line \ref{line:endbinary}.
Line \ref{line:process} updates the data structures for segments that will be used in future iterations. We prove later that this has an amortized $\O(1)$ cost.

\paragraph{Analysis of \cref{algo:sfm}.} In this section we show the following theorem.
\begin{theorem}
\label{thm:main2}
For a submodular function $f: \{0, 1\}^n \to [-1, 1]$, algorithm $\SFM(f, \eps)$ returns a random point $x \in [0, 1]^n$ with $\E[\hf(x)] \le \min_T f(T) + \eps$ and makes $\O(n/\eps^2)$ oracle calls to $f$.
\end{theorem}
\cref{thm:main2} approximately minimizes the Lovasz extension. We then use \cref{lemma:discretize} to find a discrete solution, directly implying \cref{thm:main}.

\begin{proof}[Proof of \cref{thm:main2}]
We first argue that the returned (random) point $x = \frac{1}{T+1}\sum_{i=0}^T x_i$ satisfies $\E[\hf(x)] \le \min_T f(T) + \eps$.

It suffices to show that each stochastic subgradient $\bg^{\star i}$ satisfies $\E[\| \bg^{\star i} \|_2^2] = \O(1).$ 
Then \[ 
\E[\| \bg^i \|_2^2] = s \cdot \E[\| \bg^{ \star i} \|_2^2] = \O(1) 
\] as $\bg^{\star i}$ is $s$-sparse and $s = O(\log T) = \O(1).$
Then setting $T = \O(n/\eps^2)$ suffices to apply \cref{thm:grad} with $R^2 = n$ and $B^2 = \O(1)$.

To bound of $\E \left [\| \bg^{\star i} \|_2^2 \right ]$, first define $\bz_j = \Sample(x_{k_{j+1}}, x_{k_j}, f)$ for $0 \le j \le b-1$ where $b = O(\log T)$, and $\bd$ to be the estimate of $g^0$ as defined in line \ref{line:getd}. Now, note that by Cauchy-Schwarz that
\[ 
\E \left [\|\bg^{\star i}\|_2^2 \right ] = \E\left[\Big\|\bd + \sum_{j=0}^{b-1} \bz_j\Big\|_2^2\right] \le (b+1) \cdot \E\left[\|\bd\|_2^2 + \sum_{j=0}^{b-1} \| \bz_j \|_2^2 \right] = \O(1) 
\] 
by the fact that $\E[\| \bz_j \|_2^2] = O(1)$ by \cref{lemma:sample} and $\E[\|\bd\|_2^2] \le \|g^0\|_1^2 = O(1)$ by \cref{lemma:l1}.

Now, we argue that running algorithm $\SFM(f, \eps)$ takes $\O(n/\eps^2)$ oracle calls. First, we need $O(n)$ initial oracle calls to get the gradient $g^0$ of $x_0$. Now, we have to bound the total number of oracle calls from the calls to \Sample~and \Process. 
To bound the former, we use that \Sample~is called $\O(1)$ times per iteration and only requires $\O(1)$ oracles calls (\cref{lemma:sample}). Note that $\Process(x_{k_{j+1}}, x_{k_j}, f)$ has already been called before we call $\Sample(x_{k_{j+1}}, x_{k_j}, f)$ as $k_{j+1} = k_j - 2^{\nu_2(k_j)}$ (see lines \ref{line:makesegments} and \ref{line:process}).
Therefore, the total cost of all the calls to \Sample~is $\O(T) = \O(n/\eps^2)$ as desired. 

Now we bound the total cost of oracle calls to \Process. Note that each $\bg^i$ is $1$-sparse, hence $x_i - x_{i-1}$ is 1-sparse. Therefore, for any $0 \le j \le i$, we have that $x_i - x_{i-j}$ is $j$-sparse. Thus, line \ref{line:process} takes $O(2^{\nu_2(i)})$ oracle calls to $f$ by \cref{lemma:sample}. The total number of oracle calls is thus \[ O\left( \sum_{i=1}^T 2^{\nu_2(i)} \right) = \O(T) = \O(n/\eps^2) \] as desired.
\end{proof}

\section{SFM for Functions with Sparse Minimizer}
\label{sec:sparse}

In this section, we extend our above algorithm to the case where the minimizer of $f$ is $s$-sparse and $f: \{0, 1\}^n \to [-M, M]$ is an integer valued submodular function. 
In this setting we are able eliminate the linear oracle call dependence on $n$, the dimension of the space. 
A na\"ive application of the previous algorithm runs into the following barriers: computing the initial gradient requires $O(n)$ queries and the stochastic projected gradient descent requires $O(n/\eps^2)$ iterations to converge. Previous work \cite{CLSW17} resolves the latter issue by restricting the domain to $S^s \defeq \{ x \in [0, 1]^n : \sum_i x_i \le s \}$ and arguing that the same algorithmic framework as described in \cref{sec:sfm2} extends to this setting. Here, we focus on the problem of efficiently sampling a gradient of the initial point.

There are two barriers to removing the $n$ dependence above: computing the starting subgradient and reducing the number of required iterations. To efficiently compute the starting gradient, we carefully choose a subgradient that is easier to compute.
To do so, we note that at $x_0 = 0$, every permutation corresponds to a valid subgradient. Thus, it suffices to find any permutation $P$ such that we can sample the subgradient $g^P$ with $\O(1)$ oracle calls per sample after $\O(M^2)$ preprocessing.\footnote{This can be improved to $\O(M)$ oracle calls of preprocessing. We provide a sketch in \cref{rem:m}.}

\paragraph{Efficiently sampling the initial subgradient.}
Recall that for any permutation $P$, $g^P$ is a subgradient at $0$. In this section, we give a randomized algorithm which carefully chooses a permutation $P$ and computes all nonzero coordinates of $g^P$ in $\O(M^2)$ oracle calls.

This allows us to sample future estimates of $g^P$ with variance $\O(M^2)$. In \cref{sec:derandomize} we show that we can actually \emph{derandomize} this part of the algorithm, deterministically finding a permutation $P$ which we can compute all nonzero coordinates of $g^P$ in $\O(\poly(M))$ oracle calls.

Now, we give a high level description for the algorithm. Consider any initial permutation $P_0$.
Sample a subset $S \subseteq [n]$, where each element of $[n]$ is in $S$ with probability $\frac{1}{10M}.$ 
Also, let $j$ be a coordinate such that say $g^{P_0}_j > 0$ (the $g^{P_0}_j < 0$ case is similar). 
Note that since $g^{P_0}$ has integral entries, it must be $3M$-sparse. Therefore, there is at least a $\frac{1}{10M} \left(1-\frac{1}{10M}\right)^{3M}  \ge \frac{1}{20M}$ probability that $j \in S$ and for all other coordinates $j' \in S$, we have $g^{P_0}_{j'} = 0.$ 
Condition on this event. 
Now, label the coordinates in $S$ as $j_1, j_2, \cdots, j_{|S|}$, ordered as they were originally in $P_0.$ Label the coordinates not in $S$ as $i_1, i_2, \cdots, i_{n-|S|}$, also ordered as they were originally in $P_0.$

Consider the permutation $P' = \{j_1, j_2, \cdots, j_{|S|}, i_1, i_2, \cdots, i_{n-|S|}\}$. 
Note that because $g^{P_0}_t \ge 0$ for all $t \in S$ and $g^{P_0}_j > 0$, we have that, by submodularity, $g^{P'}_{P'_t} \ge 0$ for $1 \le t \le |S|$, and that there is a $1 \le t \le |S|$ with $g^{P'}_{P'_t} > 0.$ 
Now, we can find such a coordinate in $\O(1)$ oracle calls with a binary search using \cref{lemma:sum}.
Once we find a coordinate $t$ with $g^{P'}_t > 0$, we can move that coordinate to the left of the permutation, and continue the same algorithm on the remaining elements. We can find negative elements in a similar way by moving them to the right of the permutation. Note that submodularity ensures that moving the coordinate to the left or right of the permutation never makes it zero. Repeating this process $\O(M^2)$ times gives us our permutation $P$.

After defining some additional notation we will be ready to state the algorithm.
\begin{itemize}
\item For sequences of integers $A, B$ we use $A \oplus B$ to denote concatenating $A$ and $B$. This is useful to allow us to express concatenating subsequences of permutations. For example $\{1, 3\} \oplus \{2, 4\} = \{1, 3, 2, 4\}.$
\item For a sequence $P$ and a subsequence $P'$ of $P$, the notation $P \bs P'$ means to delete the elements from $P'$ from $P$, while keeping the remaining elements in the same order as originally in $P$.
\item \textbf{Sequences $P_l$ and $P_r$.} After finding coordinates $j$ where $g^P_j$ is positive or negative, we move them to the left or right of the permutation respectively. We denote these ``fixed coordinates" as $P_l$ and $P_r$.
\item \textbf{Subsequence $S$.} This is the subset of coordinates of $P$ that we sample in an attempt to find a positive or negative coordinate.
\end{itemize}
\begin{algorithm}[h!]
\caption{$\Findperm(f)$. Takes a integer-valued submodular function $f$. Returns a pair $(P', g')$ of a permutation $P$ of $[n]$ and the associated subgradient $g' = g^{P'}$, encoded by all its $O(M)$ nonzero coordinates.}
\begin{algorithmic}[1]
\State $P \assign \{1, 2, \cdots, n \}.$ \Comment{Arbitrary initialization.}
\State $P_l, P_r \assign \emptyset$.
\For{$t=1$ to $\O(M^2)$} \label{line:start}
	\State $S$ is a random subset of $P$, where each element $j \in P$ is in $S$ with probability $\frac{1}{10M}.$ \label{line:gens}
	\State $Q \assign P_l \oplus S \oplus (P\bs S) \oplus P_r$ \Comment{Elements $S$ and $P\bs S$ are ordered as in $P$}
	\If{$\sum_{j \in S} g^Q_j > 0$} \Comment{Check for positive elements, using \cref{lemma:sum}} \label{line:check1}
		\State $x \assign \FindIndex(f, Q, S, 1).$
		\State $P_l \assign P_l \oplus \{x\}.$
		\State $P \assign P \bs \{x\}.$
		\State Go back to line \ref{line:start}.
	\EndIf
	\State $Q \assign P_l \oplus (P\bs S) \oplus S \oplus P_r$. \Comment{Elements $S$ and $P\bs S$ are ordered as in $P$}
	\If{$\sum_{j \in S} g^Q_j < 0$} \Comment{Check for negative elements, using \cref{lemma:sum}} \label{line:check2}
		\State $x \assign \FindIndex(f, Q, S, -1).$
		\State $P_r \assign \{x\} \oplus P_r.$
		\State $P \assign P \bs \{x\}.$
		\State Go back to line \ref{line:start}.
	\EndIf
\EndFor
\State \Return Permutation $P' = P_l \oplus P \oplus P_r$, with $g^{P'}$ encoded by the nonzero coordinates in $P_l$ and $P_r$.
\end{algorithmic}
\label{algo:findperm}
\end{algorithm}
\begin{algorithm}[h!]
\caption{$\FindIndex(f, P, S, b)$. Takes a integer-valued submodular function $f$, permutation $P$ of $[n]$, contiguous subset $S$ of $P$, and integer $b$ which is $\pm 1.$ Returns an index $j \in S$ such that $\sign(g^P_j) = b.$}
\begin{algorithmic}[1]
\State If $S = \{x\}$ (i.e. $S$ contains a single element), \Return $x$.
\State Split $S$ in half into subintervals $S'$ and $S''.$
\If{$\sign(\sum_{j \in S'}g^P_j) = b$} \Comment{Uses $O(1)$ oracle calls by \cref{lemma:sum}}
	\State \Return $\FindIndex(f, P, S', b)$.
\Else
	\State \Return $\FindIndex(f, P, S'', b)$.
\EndIf
\end{algorithmic}
\label{algo:findindex}
\end{algorithm}
\begin{lemma}
\label{lemma:findperm}
With high probability in $n$, \Findperm$(f)$ computes a permutation $P'$ and all the nonzero coordinates of the associated gradient $g^{P'}.$ It uses $\O(M^2)$ oracle calls to $f$.
\end{lemma}
\begin{proof}
Consider some point during the execution of $\Findperm(f)$, and the sequences $P_l, P_r, P$ at that time. Define $P' = P_l \oplus P \oplus P_r$. Our main claim is that if $g^{P'}_j \neq 0$ for some $j \in P$, then within $\O(M)$ iterations of the loop starting at line \ref{line:start}, one of line \ref{line:check1} or \ref{line:check2} will be true. To show this, let $j \in P$ be such that $g^{P'}_j \neq 0$, and without loss of generality, say $g^{P'}_j > 0.$ Because $g^{P'}$ has at most $3M$ nonzero coordinates, with probability at least \[ \left(1-\frac{1}{10M}\right)^{3M} \cdot \frac{1}{10M} \ge \frac{1}{20M} \] it will be true that $j \in S$ and for all $t \in S$ with $t \neq j$, that $g^{P'}_t = 0.$ Then by submodularity, it is clear that if we define
$Q = P_l \oplus S \oplus (P\bs S) \oplus P_r$ that $\sum_{j \in S}g^Q_j > 0$. As this happens with probability at least $\frac{1}{20M}$, it will happen w.h.p. within $\O(M)$ iterations.

Now, we must argue that $\FindIndex(f, P, S, b)$ indeed computes an index $j \in S$ such that $\sign(g^P_j) = b.$ We do the case $b = 1$ as the other is analogous. This amounts to checking that if $\sign(\sum_{j \in S} g^P_j) = b$ and $S'$ and $S''$ are subintervals of $S$ whose union is $S$, then either $\sign(\sum_{j \in S'} g^P_j) = b$ or $\sign(\sum_{j \in S''} g^P_j) = b$ but this is trivial as
\[ \sum_{j \in S} g^P_j = \sum_{j \in S'} g^P_j + \sum_{j \in S''} g^P_j. \]

To finish the proof, note that line \ref{line:check1} and \ref{line:check2} can only be true $O(M)$ times, as for any permutation $P'$ we know that $g^{P'}$ has only $3M$ nonzero coordinates. Therefore, iterating $t = \O(M^2)$ is sufficient by our main claim shown in the first paragraph. As each iteration takes $\O(1)$ function calls in $\FindIndex(f, P, S, b)$ by \cref{lemma:sum}, the total number of function calls is also $\O(M^2)$ as desired.
\end{proof}

\begin{rem}
\label{rem:m}
Here we sketch how to change \cref{algo:findperm} to improve the number of oracle calls in \cref{lemma:findperm} to $\O(M)$. This doesn't affect our main result \cref{thm:sparse} because the number of oracle calls needed to perform the projected gradient descent dominates.

If $g^P$ has exactly $t$ nonzero coordinates, then we can show that choosing $S$ as a random subset of $P$, where each element $j \in P$ is in $S$ with probability $p$ for $\frac{1}{20t} \le p \le \frac{1}{10t}$ (analogous to line \ref{line:gens} of \cref{algo:findperm}) will isolate some nonzero coordinate of $g^P$ with at least constant probability. This is because each nonzero coordinate of $g^P$ has at least a $p \cdot (1-p)^t \ge \frac{1}{100t}$ chance of being isolated. Unioning over all $t$ nonzero coordinates (which correspond to disjoint events) shows that there is at least a $\frac{1}{100}$ probability of some nonzero coordinate being isolated. Therefore, running this $\O(1)$ times isolates some coordinate w.h.p.

As we do not know $t$, we iterate over guesses for $t$, i.e. set $p = 2^{-i}$ for for $0 \le i \le O(\log M)$ and run the process described in the above paragraph for each value of $p$.
\end{rem}

\paragraph{Projecting onto $S^s$.} The $\ell_2$ projection onto $S^s$ can be computed as follows. This was stated in \cite{CLSW17}.
\begin{lemma}
\label{lemma:projss}
For $s \ge 0$ let $S^s = \{x \in [0, 1]^n : \sum_i x_i \le s \}.$ For any $y \in \R^n$, we have that the point $z = \proj(y, S^s)$ is given by
$z_i = \text{median}(0, 1, y_i-\lambda)$, where $\lambda$ is the smallest nonnegative real number such that $\sum_i z_i \le s.$
\end{lemma}
Note that \cref{lemma:projss} shows that the permutation $P$ consistent with $y$ is also consistent with $\proj(y, S^s).$

\paragraph{Algorithm description and analysis.} After finding a permutation $P$ where we can efficiently sample $g^P$, we set $x_0 = 0$ (the origin), which is consistent with every permutation, and run a variation of \cref{algo:sfm} where we project onto $S^s.$ We need the following variation on \cref{lemma:sample} to deal with the projections. \cref{lemma:sample2} was also argued in \cite{CLSW17}. A proof sketch is provided in \cref{sec:proofs}.
\begin{restatable}{lemma}{restatesampletwo}
\label{lemma:sample2}
Let $f:\{0, 1\}^n \to [-M, M]$ be a submodular function with Lovasz extension $\hf$. Let $g$ denote the subgradients of $\hf$. Let $x, y \in [0, 1]^n$ be vectors and let $P_x$ and $P_y$ be permutations consistent with $x, y$ respectively. Assume that we can transform $P_x$ into $P_y$ by deleting $k$ elements from $P_x$ and inserting them back in other locations. There is a data structure which after $O(k)$ calls to $f$ of preprocessing supports the following: sample a $1$-sparse random variable $\bz$ with $\E[\bz] = g(y)-g(x)$ and $\E[\|\bz\|_2^2] = O(1)$ in $\O(1)$ calls to $f$. Preprocessing is called through $\Process(x, y, f)$, and the sampling is called through $\Sample(x, y, f).$
\end{restatable}
In other words, we don't necessarily need for $y-x$ to be $k$-sparse as in \cref{lemma:sample}; it suffices for their consistent permutations to only ``differ" by $k$ moves. This essentially follows from the fact that $g(y)$ only depends on $P_y$.

At this point we are ready to state our algorithm.
\begin{algorithm}[h!]
\caption{$\SparseSFM(f, s, \eps)$. Takes a submodular function $f: \{0,1\}^n \to [-M, M]$ with an $s$-sparse minimzer and returns a random point $x \in [0, 1]^n$ with $\E[\hf(x)] < \min_T f(T) + 1$.}
\begin{algorithmic}[1]
\State $T \assign \O(sM^2).$
\State $x_0 \assign 0 \in \R^n.$
\State $(P_0, g^0) \assign \Findperm(f).$
\For{$i=1$ to $T$}
	\State $b \assign $ the number of $1$ bits in the binary representation of $i-1$. \label{line:startbinary2}
	\State $k_0 \assign i-1$, $k_{j+1} \assign k_j - 2^{\nu_2(k_j)}$ for $0 \le j \le b-1$. \Comment{$\nu_2(y)$ is the maximum integer $t$ such that $2^t$ divides $y$}
	\State $\bd \assign \|g^0\|_1 \cdot \sign(g^0_k)e_k$ with probability $|g^0_k|/\|g^0\|_1.$ \Comment{Estimate of $g^0$} \label{line:getd2}
	\State $\bg^{\star i} \assign \bd + \sum_{j=0}^{b-1} \Sample(x_{k_{j+1}}, x_{k_j}, f).$
	\State Let $c_1, c_2, \cdots, c_w$ be the nonzero coordinates of $\bg^{\star i}$ \Comment{$w \le b+1$}
	\State $\bg^i \assign w \cdot \bg^{\star i}_{c_k} \cdot e_{c_k}$ with probability $\frac{1}{w}.$ \Comment{$\bg^i$ is $1$-sparse}
	\State $x_i \assign \proj(x_{i-1} - \eta \bg^i, S^s).$ \label{line:endbinary2}
	\State $\Process(x_{i-2^{\nu_2(i)}}, x_i, f).$ \label{line:process2}
\EndFor
\State \Return $\frac{1}{T+1} \sum_{i=0}^T x_i.$
\label{algo:sparsesfm}
\end{algorithmic}
\end{algorithm}
\begin{theorem}
\label{thm:sparse2}
For submodular function $f: \{0, 1\}^n \to [-M, M]$, algorithm $\SparseSFM(f, s, \eps)$ returns random point $x \in [0, 1]^n$ with $\E[\hf(x)] < \min_T f(T) + 1$ using $\O(sM^2)$ oracle calls to $f$.
\end{theorem}
\begin{proof}[Proof sketch]
Copy the proof of \cref{thm:main2}, replacing \cref{lemma:sample} with \cref{lemma:sample2}. $T = \O(sM^2)$ suffices as we can set $\eps = \frac{1}{2}$, and $B^2 = \O(M^2)$, and $R^2 = s$ in \cref{thm:grad}.
\end{proof}
It is direct to see that \cref{thm:sparse2} implies \cref{thm:sparse}.

\section{SFM over Domain $[k]^n$}
\label{sec:sfmk}

Previous work \cite{Bach19} has considered more general domains for submodular functions, instead of the standard $\{0,1\}^n$. Such a domain that the definition of submodularity can be extended to is functions $f: [k]^n \to \R$. We call a function $f: [k]^n \to \R$ submodular if for all $x, y \in [k]^n$ we have that
\[ f(x)+f(y) \ge f(\max\{x, y\})+f(\min\{x, y\}) \] where $\max$ and $\min$ are applied entry-wise. We assume without loss of generality that $f((1, \cdots, 1)) = 0.$

This definition can be further extended to functions over continuous domains. This has also been considered in previous work. We call a function $f: [0, 1]^n \to \R$ submodular if for all $i, j \in [n]$ with $i \neq j$ we have that $\frac{\partial^2 f}{\partial x_i \partial x_j} \le 0$, i.e. all mixed partials are non-positive everywhere.

In this section we show how to obtain algorithms for minimizing submodular functions in each setting that make a number of oracle calls nearly linear in $n$.

\subsection{Preliminaries}
We start by providing the necessary definitions for this section.

\paragraph{General notation.} Define $[k] \defeq \{1, 2, \dots, k\}$ and $[k]^n \defeq \{(x_1, x_2, \cdots, x_n) : x_i \in [k] \forall i \}$.

\paragraph{Continuous extension.} Here, we define the continuous extension for submodular functions $f : [k]^n \to \R$. All the results below were proven by Bach \cite{Bach19}. We first define its domain.
\begin{definition}[Domain of continuous extension]
\label{def:domain}
Let $f : [k]^n \to \R$ be a submodular function. Define the set $H_k \defeq \{x \in [0, 1]^{k-1} : x_1 \ge x_2 \ge \cdots \ge x_{k-1} \}.$ The set $H_k^n \defeq \overbrace{H_k \times \cdots \times H_k}^{n \text{ times}}$ will be the domain of the continuous extension of $f$.
\end{definition}
For a point $x = (x^1, x^2, \cdots, x^n) \in H_k^n$, we define $x_{a, b} \defeq x^a_b.$

We define a permutation consistent to a point $x \in H_k^n$. This is a generalization of the situation for submodular functions over $\{0, 1\}^n$.
\begin{definition}[Associated permutation to a submodular function]
\label{def:perm}
An \emph{associated permutation} to a point $x = (x^1, x^2, \cdots, x^n) \in H_k^n$, denoted $(P, Q)$, is a permutation of $[n] \times [k-1]$, given by $(P_1, Q_1), (P_2, Q_2), \cdots, (P_{(k-1)n}, Q_{(k-1)n})$, which satisfies $x^{P_i}_{Q_i} \ge x^{P_{i+1}}_{Q_{i+1}}$ for $(k-1)n > i \ge 1.$
\end{definition}
We now define the continuous extension of a submodular function.
\begin{definition}[Continuous extension of a submodular function]
\label{def:continuous}
Let $f : [k]^n \to \R$ be a submodular function. We define the continuous extension $\hf: H_k^n \to \R$ of $f$ as follows. For a point $x = (x^1, x^2, \cdots, x^n) \in H_k^n$, let $(P, Q)$ be an associated permutation to $x$. Define the sequence of points $S_0, S_1, \cdots, S_{(k-1)n} \in [k]^n$ as $S_0 = (1, 1, \cdots, 1)$ and $S_i = S_{i-1} + e_{P_i}$ for $(k-1)n \ge i \ge 1.$ Then we define
\[ \hf(x) = x^{P_{(k-1)n}}_{Q_{(k-1)n}} f(S_{(k-1)n}) + \sum_{i=1}^{(k-1)n-1}(x^{P_i}_{Q_i}-x^{P_{i+1}}_{Q_{i+1}})f(S_i)\]
\end{definition}
It is direct to see that \cref{def:continuous} essentially reduces to the Lovasz extension in the case $k = 2$.

We now give an example illustrating \cref{def:continuous}.
\begin{example}
Consider the following submodular function $f:[3]^2 \to \R.$
\begin{align*}
&f(1, 1) = 0, f(1, 2) = 1, f(1, 3) = 2 \\
&f(2, 1) = 1, f(2, 2) = 2, f(2, 3) = 2 \\
&f(3, 1) = 0, f(3, 2) = 1, f(3, 3) = 0.
\end{align*}
Consider the following point in $H_3^2$: $x = (x^1, x^2) = ((0.6, 0.3), (0.5, 0.1)).$ The permutation $(P, Q)$ consistent with $x$ is
$\{ (1, 1), (2, 1), (1, 2), (2, 2) \}$ as $x^1_1 \ge x^2_1 \ge x^1_2 \ge x^2_2.$ This lets us compute that \[ S_0 = (1, 1), S_1 = (2, 1), S_2 = (2, 2), S_3 = (3, 2), S_4 = (3, 3). \] Therefore, we have that
\begin{align*}
\hf(x) &= 0.1 \cdot f(S_4) + 0.2 \cdot f(S_3) + 0.2 \cdot f(S_2) + 0.1 \cdot f(S_1) \\
	   &= 0.1 \cdot f(3, 3) + 0.2 \cdot f(3, 2) + 0.2 \cdot f(2, 2) + 0.1 \cdot f(2, 1) \\
	   &= 0.1 \cdot 0 + 0.2 \cdot 1 + 0.2 \cdot 2 + 0.1 \cdot 1 = 0.7.
\end{align*}
\end{example}
The following properties of the continuous extension are known. See Sections 3, 4, 5 in \cite{Bach19} for proofs.

\begin{theorem}[Properties of the continuous extension]
\label{def:properties}
Let $f: [k]^n \to \R$ be a submodular function, and let $\hf$ be its continuous extension. Then we have that
\begin{itemize}
\item $\hf$ is convex.
\item For $S = (s_1, s_2, \cdots, s_n) \in [k]^n$, if we define $x^i \in H_k$ as $x^i = \sum_{j=1}^{s_i-1} e_j$, then for $x = (x^1, x^2, \cdots, x^n) \in H_k^n$ we have that $\hf(x) = f(S).$
\item We have that $\min_{x \in H_k^n} \hf(x) = \min_{S \in [k]^n} f(S).$
\end{itemize}
Additionally, the vector $g(x) \in \R^{n \times (k-1)}$ defined by $g(x)_{P_j, Q_j} \defeq f(S_j) - f(S_{j-1})$ is a subgradient of the continuous extension, where the $S_j$ are defined as in \cref{def:continuous}.
\end{theorem}

\subsection{SFM over $[k]^n$}
In this section we sketch an algorithm and analysis for submodular function minimization of functions $f: [k]^n \to [-1, 1].$ Precisely we show the following result.
\begin{theorem}
\label{thm:maink}
Given a submodular function $f: [k]^n \to [-1, 1]$ and an $\eps > 0$, we can compute a random point $x \in [k]^n$ with \[ \E[f(x)] \le \min_{y \in [k]^n} f(y) + \eps \] in $\O(nk^4/\eps^2)$ calls to an oracle for $f$.
\end{theorem}
As the algorithm is extremely similar to those presented in \cref{sec:sfm2} we simply state the analogues of the lemmas we must show and how they imply the result. Precisely we need the analogues of \cref{lemma:l1}, \cref{lemma:change}, and \cref{lemma:sum} for the continuous extension which was defined in \cref{sec:prelim}. The proofs are analogous to those of \cref{lemma:l1}, \cref{lemma:change}, and \cref{lemma:sum} which were given in \cite{CLSW17}.

Before stating the lemmas, we remark that the setup and notation we are using is as in \cref{def:continuous}.
\begin{restatable}{lemma}{restatelonek}
\label{lemma:l1k}
For a submodular function $f: [k]^n \to [-M, M]$, all subgradients $g$ of the continuous extension satisfy $\|g(x)\|_1 \le 4M(k-1).$
\end{restatable}
\begin{restatable}{lemma}{restatechangek}
\label{lemma:changek}
Let $x = (x^1, \cdots, x^n), y = (y^1, \cdots, y^n) \in H_k^n$, and $d = (d^1, \cdots, d^n) \in \R_{\ge 0}^{n \times (k-1)}$ be such that $y = x+d$ (respectively $y = x-d$). For all $i$ such that $d^i = 0$ and $j \in [k-1]$ we have that $g(x)_{i, j} \ge g(y)_{i, j}$ (respectively $g(x)_{i,j} \le g(y)_{i,j}$).
\end{restatable}
\begin{restatable}{lemma}{restatesumk}
\label{lemma:sumk}
Let $x \in H_k^n$ and let $(P, Q)$ be the permutation consistent with $x$. Then we have for any integers $1 \le a \le b \le n(k-1)$ that \[ \sum_{i=a}^b g(x)_{P_i,Q_i} = f(S_b) - f(S_{a-1}), \] where the $S_j$ are defined as in \cref{def:continuous}.
\end{restatable}
We prove these in \cref{sec:proofs}.

Additionally, as the algorithm we intend to use is projected subgradient descent, we must be able to project onto $H_k^n.$ Projecting onto $H_k$ is simply an isotonic regression, which can be done via the pool-adjacent-violators algorithm \cite{BC90}.

Finally, we need the analogue of \cref{lemma:sample}. The proof is analogous to that of \cref{lemma:sample} and we provide a sketch in \cref{sec:proofs}.
\begin{restatable}{lemma}{restatesamplek}
\label{lemma:samplek}
Let $f:[k]^n \to [-1, 1]$ be a submodular function with continuous extension $\hf$. Let $g$ denote the subgradients of $\hf$. Let $x = (x^1, \cdots, x^n), y = (y^1, \cdots, y^n) \in H_k^n$ be vectors. Let $d = (d^1, \cdots, d^n) \in \R^{n \times (k-1)}$ be the vector such that $d = y-x$, and say that there are $\ell$ indices $i \in [n]$ such that $d^i \neq 0.$ There is a data structure which after $O(\ell k)$ calls to $f$ of preprocessing, supports the following: sample a $1$-sparse random variable $\bz$ with $\E[\bz] = g(y)-g(x)$ and $\E[\|\bz\|_2^2] = O(k^2)$ in $\O(1)$ calls to $f$. Preprocessing is called through $\Process(x, y, f)$, and the sampling is called through $\Sample(x, y, f).$
\end{restatable}
The condition $\E[ \| \bz \|_2^2] = O(k^2)$ comes from the fact that our algorithm will satisfy $\E[|\bz|_2^2] \le O\left(\max_x \| g(x) \|_1^2\right) \le O(k^2)$ by \cref{lemma:l1k}.

\begin{proof}[Proof of \cref{thm:maink}]
We use basically the exactly same algorithm as \cref{algo:sfm}, except with the projections in line \ref{line:endbinary} replaced with projections onto $H_k^n$. We can compute that (in the language of \cref{thm:grad}) we have that $R^2 = nk$ and $B^2 = O(k^2)$, hence we have expected error $\eps$ in $O(R^2B^2/\eps^2) = O(nk^3/\eps^2)$ steps. By \cref{lemma:samplek}, the total number of calls in the procedure corresponding to line \ref{line:process} of \cref{algo:sfm} will take on average $\O(k)$ times the iteration count. Therefore, the total number of function calls to $f$ is $\O(k \cdot nk^3/\eps^2) = \O(nk^4/\eps^2)$ as desired.
\end{proof}

\subsection{SFM for continuous functions}
\label{sec:nonconvex}
In this section we prove \cref{thm:nonconvex}. Our setup is the following: we have a submodular function $f: [0,1]^n \to \R$, and we wish to approximately minimize $f$. Our algorithms are in terms of the $L^\infty$-Lipschitz constant of $f$, which we denote as $L$.

Our algorithm is simple: we essentially just discretize $f$ and use \cref{thm:maink}. Specifically, define $k = \frac{2L}{\eps}$, and define the function $f':[k]^n \to \R$ as $f'(x) = f(x/k)$ for $x \in [k]^n \subseteq \R^n$. Note that it is clear that \[ \min_{x \in [k]^n} f'(x) \le \min_{x \in [0,1]^n} f(x) + \frac{L}{k} \le \min_{x \in [0,1]^n} f(x) + \eps/2 \] by our choice of $k$. We can also verify that $f'$ is submodular. Therefore, it suffices to minimize $f'$ within $\eps/2$. Without loss of generality, we assume $f'(1, 1, \cdots, 1) = 0.$

We can almost directly apply \cref{thm:maink}, except that the range of $f'$ is not $[-1, 1]$, and is instead $[-L, L]$. This change multiplies the $B^2$ term in our application of \cref{thm:grad} by $L^2$ (so that $B^2 = O(k^2L^2)$), giving a total complexity of $\O(nk^4L^2/\eps^2) = \O(nL^6/\eps^6)$ oracle calls to $f$ as desired.

We now formally give the proof. It is essentially as described above.
\begin{proof}[Proof of \cref{thm:nonconvex}]
Define $k = \frac{2L}{\eps}$, and define the function $f':[k]^n \to \R$ as $f'(x) = f(x/k)$ for $x \in [k]^n \subseteq \R^n$. We have that \[ \min_{x \in [k]^n} f'(x) \le \min_{x \in [0,1]^n} f(x) + \frac{L}{k} \le \min_{x \in [0,1]^n} f(x) + \eps/2.\] Therefore, it suffices to minimize $f'$ to additive $\eps/2.$

To do this, we use the same algorithm as in the proof of \cref{thm:maink}. As in the notation of \cref{thm:grad}, we have that $R^2 = nk$, and $B^2 = O(k^2L^2)$, where the extra factor of $L^2$ comes from the fact that the range of $f'$ is $O(L)$. Thus, the expected error is $\eps/2$ in $O(R^2B^2/\eps^2) = O(nk^3L^2/\eps^2)$ iterations. By the same argument as in the proof of \cref{thm:maink}, we have that on average, each iteration requires $\O(k)$ function calls. The total number of calls is therefore $\O(nk^4L^2/\eps^2) = \O(nL^6/\eps^6)$ function calls by our choice of $k$.
\end{proof}

\subsection*{Acknowledgements}

We thank Deeparnab Chakrabarty, Yin Tat Lee, Sahil Singla, Kevin Tian, and Sam Chiu-wai Wong for helpful conversations. 

{\small
\bibliographystyle{alpha}
\bibliography{refs}}

\begin{appendix}
\section{Nonconstructive derandomization of \cref{algo:sparsesfm}}
\label{sec:derandomize}
In this section we explain how we can nonconstructively derandomize \cref{algo:sparsesfm}. In other words, we sketch a deterministic algorithm such that given an integer-valued submodular function $f:\{0,1\}^n \to [-M,M]$ finds a permutation $P$ such that we can compute all nonzero coordinates of $g^P$ in $\O(\poly(M))$ oracle calls.

Now we explain the main idea behind the derandomization. Let's consider the case in \cref{algo:sparsesfm} when $P_l = P_r = \emptyset$ at the start. Let $S$ be the random subset generated in line \ref{line:gens}, where each element in $P$ is in $S$ with probability $\frac{1}{10M}.$ As explained in the proof of \cref{lemma:findperm}, this choice of $S$ allows to make progress as long as there is exactly one index $i \in S$ such that $g^P_i \neq 0$, and for all other indices $j \in S$ with $j \neq i$ we have that $g^P_j = 0.$ As long as $g^P \neq 0$, the probability of this occuring is at least \[ \frac{1}{10M} \cdot \left(1-\frac{1}{10M}\right)^{3M} \ge \frac{1}{20M}. \] Imagine randomly generating $T = 500M^2\log n$ such sets $S_1, \cdots, S_T$. The probability that there is no set $S_j$ satisfying the desired property of exactly one index $i \in S_j$ with $g^P_i \neq 0$ is at most $\left(1 - \frac{1}{20M}\right)^{500M^2 \log n} \le n^{-20M}.$ Note that because $\|g^P\|_1 \le 3M$ by \cref{lemma:l1}, there are at most $(2n)^{3M}$ distinct possible subgradients $g^P.$ Because \[ (2n)^{3M} \cdot n^{-20M} < 1, \] a union bound tells us that there exist \emph{deterministic} sets $S_1, \cdots, S_T$ that the algorithm can precompute independent of $f$ such that for any nonzero subgradient $g^P$ there is a set $S_j$ for $1 \le j \le T$ such that for exactly index $i \in S_j$ we have $g^P_i \neq 0$, as desired.

We now formally state this discussion as a lemma.
\begin{lemma}
\label{lemma:derandomize}
There is a deterministic algorithm which give an integer-valued submodular function $f:\{0,1\}^n \to [-M,M]$ computes a permutation $P$ and finds all nonzero entries of $g^P$ in $\O(M^3)$ oracle calls.
\end{lemma}
\begin{proof}
This proof essentially follows the above discussion. Define $T = 500M^2 \log n$. Our goal is to deterministically construct sets $S_1, S_2, \dots, S_T \subseteq [n]$ such that for any nonzero vector $g \in \R^n$ with integer entries and $\|g\|_1 \le 3M$, that for some $1 \le j \le T$, we have that there is exactly one element $i \in S_j$ such that $g_i \neq 0.$ With this construction, we can simply copy the proof of \cref{lemma:findperm}, using that all subgradients $g$ of the Lovasz extension have integer entries and $\ell_1$ norm at most $3M$. We focus on this goal in the remainder of the proof.

Randomly generate subsets $S_1, \cdots, S_T \subseteq [n]$ as follows: each $S_j$ is such that each $i \in P$ is independently in $S_j$ with probaiblity $\frac{1}{10M}.$ Let $g \in \R^n$ be a nonzero vector with integer entries and $\|g\|_1 \le 3M$. We now bound the probability for some $S_j$ we have that there is exactly one element $i \in S_j$ with $g_i \neq 0.$ For a fixed $j$, the probability that $S_j$ satisfies this property is at least \[ \frac{1}{10M} \cdot \left(1 - \frac{1}{10M}\right)^{3M} \ge \frac{1}{20M}. \] Therefore, the probability that no $S_j$ satisfy the desired property is at most
\[ \left(1 - \frac{1}{20M}\right)^T \le n^{-20M}. \]

Our next goal is to count the number of possible distinct vectors $g \in \R^n$ with integer entries and $\|g\|_1 \le 3M.$ A direct counting argument easily shows that the number of such vectors $g$ is at most $\sum_{k=0}^{3M} (2n)^k \le 2(2n)^{3M}.$ Therefore, by a union bound (as $2(2n)^{3M} \cdot n^{-20M} < 1$) there exists sets $S_1, \dots, S_T \subseteq [n]$ such that for all vectors $g \in \R^n$ with integer entries and $\|g\|_1 \le 3M$ that there is a $j$ such that there is exactly one element $i\in S_j$ with $g_i \neq 0.$ We can find such sets $S_1, \dots, S_T$ just by brute forcing over all possibilities: these are independent of $f$, so they do not cost oracle calls to compute.
\end{proof}
It would be interesting to give a polynomial time algorithm to deterministically construct sets $S_1, S_2, \dots, S_T$ as described in \cref{lemma:derandomize}.
\begin{rem}
We can improve the number of oracle calls in \cref{lemma:derandomize} to $\O(M^2)$ and the value of $T$ in the proof to $\O(M)$ using the same technique as in \cref{rem:m}.
\end{rem}

\section{Additional Proofs}
\label{sec:proofs}

\restatesample*
\begin{proof}
Let $d = y-x$. We first argue that it suffices to consider the case where $d$ either has all nonnegative or all nonpositive coordinates.
To this end, let $d^+, d^- \in \R^n$ be the positive and negative parts of $d$, precisely defined as
\[ d^+_i = \max(0, d_i) \text{ and } d^-_i = \min(0, d_i) \forall 1 \le i \le n. \] Write
\[ g(y)-g(x) = \left(g(x+d^++d^-)-g(x+d^+)\right) + \left(g(x+d^+)-g(x)\right). \] To sample the estimate $\bz$ for $g(y)-g(x)$, we instead sample $\bz_1$ for $\left(g(x+d^++d^-)-g(x+d^+)\right)$ and $\bz_2$ for $\left(g(x+d^+)-g(x)\right)$, and set $\bz$ to be either $2\bz_1$ or $2\bz_2$, each with probability $\frac12.$ It is clear that if both $\E[\| \bz_1 \|_2^2] = O(1)$ and $\E[\| \bz_2 \|_2^2] = O(1)$, then $\E[\| \bz \|_2^2] = O(1).$ This shows that we can reduce to the case where either $d$ has all nonnegative or nonpositive coordinates.

By symmetry, we only consider the case where $d$ has all nonnegative coordinates. Let $y = x+d.$ Let $P_x$ be the permutation consistent with $x$, and let $P_y$ be the permutation consistent with $y$. Note that because $d$ is $k$-sparse, one can transform permutation $P_x$ into $P_y$ deleting $k$ elements from $P_x$ and inserting them back. Therefore, there exist subsets $I_1, I_2, \cdots, I_{2k} \subset [n]$ that are intervals in both $P_x$ and $P_y$.

The phase $\Process(x, y, f)$ then proceeds computing $D_t \defeq \sum_{j \in I_t} (g(y)_j-g(x)_j)$ for all $1 \le t \le 2k.$ By \cref{lemma:sum} this requires $O(k)$ queries to $f$. Note that each for each $j \in I_t$, the terms $g(y)_j-g(x)_j$ are all the same sign by \cref{lemma:change}, hence $\sum_{t=1}^{2k} |D_t| = |g(y)-g(x)|_1.$

$\Sample(x, y, f)$ proceeds as follows. Choose an interval $I_t$ proportional to $|D_t|.$ Let $I$ be the interval that is chosen. Split this interval in half into two intervals $I'$ and $I''.$ Compute the sums $D' = \sum_{j \in I'} (g(y)_j-g(x)_j)$ and $D'' = \sum_{j \in I''} (g(y)_j-g(x)_j).$ Sample one of $I'$ and $I''$ proportional to $D'$ and $D''$, respectively. Now continue recursively. When the interval is size $1$, say containing the element $j$, return the vector $\bz = \| g(y)-g(x) \|_1 \cdot \sign(g(y)_j-g(x)_j) \cdot e_j$. By \cref{lemma:sum} this phase takes $\O(1)$ queries to $f$ and returns a $1$-sparse estimate $\bz$ for $g(y)-g(x).$ We can check by the construction that $\E[\bz] = g(y)-g(x)$, and \[ \E[ \| \bz \|_2^2] \le \| g(y)-g(x) \|_1^2 = O(1) \] by \cref{lemma:l1}.
\end{proof}

\restatesampletwo*
\begin{proof}[Proof sketch]
It is direct to see that if we can get from $P_x$ to $P_y$ by deleting $k$ elements from $P_x$ and inserting them back in other positions, then there exist points $x', y' \in [0,1]^n$ where all coordinates of $x'$ are distinct and all coordiantes of $y'$ are distinct, such that $P_x$ is consistent with $x'$, $P_y$ is consistent with $y'$, and $y'-x'$ is $k$-sparse. Now use the proof of \cref{lemma:change} above on the points $x'$ and $y'$.
\end{proof}

\restatelonek*
\begin{proof}
Let $g$ be the gradient at a point $x$. We prove that for $1 \le j < k$ we have that $\sum_{i=1}^n |g_{i,j}| \le 4M.$ Summing over all $j$ then gives us that $\|g\|_1 \le 4M(k-1).$ We first bound the sum of the positive entries of $g$, i.e. we show that $\sum_{i=1}^n \max(0,g(i,j)) \le 2M$ for any $j \in [k-1].$ An analogous argument will show that $\sum_{i=1}^n \min(0,g(i,j)) \ge -2M$, which together is sufficient.

Fix $j \in [k-1].$ Let $(P,Q)$ be the permutation corresponding to our point $x$. Without loss of generality, assume that $g_{1,j} \ge g_{2,j} \ge \cdots \ge g_{n,j}.$ Let $a_1, a_2, \dots, a_n$ be such that $(P_{a_y}, Q_{a_y}) = (y,j).$ Let the $S_i$ be defined as in \cref{def:continuous}. Note that $a_1 \le \cdots \le a_n$ by our assumption. Let $X = \{i \in [n] : g_{i,j} > 0\}$, let $t = |X|$ and let $i_1 \le \cdots \le i_t$ be the elements of $X$. Define $v_0 = (j-1, j-1, \cdots, j-1) \in \R^n.$ For $1 \le y \le t$, define $v_y = v_{y-1} + e_{i_y}.$ Note that by the definition of submodularity over $[k]^n$ that $f(v_y)-f(v_{y-1}) \ge f(S_{a_{i_y}})-f(S_{a_{i_y}-1}) = g_{i_y,j}$. Therefore, we have that
\[ 2M \ge f(v_t)-f(v_0) = \sum_{y=1}^t f(v_y)-f(v_{y-1}) \ge \sum_{y=1}^t g_{i_y,j} \] as desired.
\end{proof}

\restatechangek*
\begin{proof}
We only prove the case $y = x+d$, as the $y = x-d$ case is analogous. Let $(P, Q)$ be the permutation corresponding to $x$ and $(P',Q')$ is the permutation corresponding to $y$. Let $S_i$ be the sets defined in \cref{def:continuous} for $x$, and let $S_i'$ be the sets for $y$.

Let $a, b$ be such that $(i, j) = (P_a, Q_a)$ and $(i, j) = (P_b', Q_b')$. Then \cref{def:properties} tell us that $g(x)_{i,j} = f(S_a)-f(S_{a-1})$ and $g(y)_{i,j} = f(S_b')-f(S_{b-1}')$, where $S_a = S_{a-1}+e_i$ and $S_b' = S_{b-1}' +  e_i$ as defined in \cref{def:continuous}. We will show that $(S_{a-1})_t \le (S_{b-1}')_t$ for all indices $t \in [n]$, and $(S_{a-1})_i = (S_{b-1}')_i.$ Then the inequality \[ f(S_a)-f(S_{a-1}) = f(S_{a-1}+e_i)-f(S_{a-1}) \ge f(S_{b-1}'+e_i)-f(S_{b-1}') = f(S_b')-f(S_{b-1}') \] by the definition of submodularity over $[k]^n.$

To argue that $S_{b-1}' \ge S_{a-1}$ coordinate-wise, note that because $y = x+d$ and $d \ge 0$ that we can get from $(P,Q)$ to $(P',Q')$ by moving some pairs $(P_t,Q_t)$ to the left (where the ``left" has the larger elements) but without touching any $(P_t,Q_t)$ with $P_t=i$ as $d^i = 0.$ By the definition of the $S_i$ in \cref{def:continuous}, we can now directly check that $S_{b-1}' \ge S_{a-1}$ entry-wise, and $(S_{a-1})_i = (S_{b-1}')_i$ as desired.
\end{proof}

\restatesumk*
\begin{proof}
This follows immediately from the definition of $g$. By \cref{def:properties} we have that $g(x)_{P_i,Q_i} = f(S_i)-f(S_{i-1}).$ Therefore, \[ \sum_{i=a}^b g(x)_{P_i,Q_i} = \sum_{i=a}^b f(S_i)-f(S_{i-1}) = f(S_b) - f(S_{a-1}).\]
\end{proof}

\restatesamplek*
\begin{proof}[Proof sketch]
Using the same technique as in \cref{lemma:sample} we reduce to the case where $d$ has all non-negative coordinates, so that $y=x+d$. Because there are $\ell$ indices $i \in [n]$ such that $d^i \neq 0$, we can transform the associated permutation $(P,Q)$ of $x$ to the associated permutation $(P',Q')$ of $y$ by deleting and reinserting $k\ell$ elements, corresponding to $k$ coordinates per each $i$ with $d^i \neq 0$, and there are at most $\ell$ such indices $i$.

By submodularity, we can construct $O(k\ell)$ intervals where $g(y)-g(x)$ is either all positive or all negative. We can preprocess these intervals in $O(k\ell)$ oracle calls as done in \cref{lemma:sample}. Afterwards, we can sample $1$-sparse estimates to $g(y)-g(x)$ in $\O(1)$ queries. As in \cref{lemma:sample} our estimate will satisfy $\E[\|\bz\|_2^2] = \|g(y)-g(x)\|_1^2 = O(k^2)$ by \cref{lemma:l1k}.
\end{proof}
\end{appendix}

\end{document}